\documentclass[10pt]{article}

\usepackage{graphicx}%
\usepackage{amsmath,amssymb,amsthm,upref,bm}%
\usepackage{mathrsfs}

\DeclareMathAlphabet{\varmathbb}{U}{pxsyb}{m}{n}

\newtheorem{lemma}{Lemma}%
\newtheorem{theorem}{Theorem}%

\newcommand{\pd}[3][]{\mathchoice{\raise-0.5pt\hbox{$\partial$}%
\vphantom{\partial}_{\mkern-1.5mu#2}^{\mkern0.4mu#1}\mkern0.3mu}%
{\raise-0.5pt\hbox{$\partial$}%
\vphantom{\partial}_{\mkern-1.5mu#2}^{\mkern0.4mu#1}\mkern0.3mu}%
{\raise-0.5pt\hbox{$\scriptstyle\partial$}%
\vphantom{\partial}_{\mkern-1.7mu#2}^{\mkern0.1mu#1}\mkern0.1mu}%
{\raise-0.5pt\hbox{$\scriptscriptstyle\partial$}%
\vphantom{\partial}_{\mkern-1.7mu#2}^{\mkern0.1mu#1}\mkern0.1mu}#3}
 
\newcommand{\D}{\mathrm{d}\kern0.2pt}%
\newcommand{\ii}{\kern0.05em\mathrm{i}\kern0.05em}%

\begin{document}

\baselineskip=4.4mm

\makeatletter

\title{\bf A Faber--Krahn inequality for \\ indented and cut membranes}

\author{Nikolay Kuznetsov}

\date{}

\maketitle

\vspace{-10mm}

\begin{center}
Laboratory for Mathematical Modelling of Wave Phenomena, \\ Institute for Problems
in Mechanical Engineering, Russian Academy of Sciences, \\ V.O., Bol'shoy pr. 61,
St. Petersburg 199178, Russian Federation \\ E-mail: nikolay.g.kuznetsov@gmail.com
\end{center}

\begin{abstract}
In 1960, Payne and Weinberger proved that among all domains that lie within a wedge
(an angle whose measure is less than or equal to $\pi$), and have a given value of a
certain integral the circular sector has the lowest fundamental eigenvalue of the
Dirichlet Laplacian. Here, it is shown that an analogue of this assertion is true
for domains with a cut and for indented domains; that is, for those located in a
reflex angle (its measure is between $\pi$ and $2 \pi$).
\end{abstract}

\setcounter{equation}{0}

\section{Introduction}

Isoperimetric inequalities for eigenvalues of the Laplacian have its roots in the
work of Lord Rayleigh presented in the first volume of his monograph \textit{The
Theory of Sound} \cite{R}. It was found that the normal modes and proper frequencies
characterizing the vibrations of a fixed, homogeneous, elastic membrane are
determined by the eigenvalue problem for the Dirichlet Laplacian on a plane, bounded
domain; see \cite{KS} for a review and historical remarks.

Indeed, let $D \subset \mathbb{R}^2$ be a bounded domain with a piecewise smooth
boundary (the membrane at rest coincides with $D$). If for some real $\lambda$ the
boundary value problem
\[ u_{xx} + u_{yy} + \lambda u = 0 \ \ \mbox{in} \ D , \quad u = 0 \ \ \mbox{on} \ 
\partial D
\]
has a non-trivial solution continuous on $\bar D$ and belonging to $C^2 (D)$, then
$\lambda$ and $u$ are a Dirichlet eigenvalue of $D$ and the corresponding
eigenfunction respectively. The sequence of Dirichlet eigenvalues is positive and
the squares of the membrane eigenfrequencies are proportional to them. This sequence
is characterised by the max-min principle (see, for example, \cite{AB,B}), according
to which the lowest eigenvalue $\lambda_1 (D)$ is simple and the corresponding
eigenfunction $u_1$ can be taken to be positive in the interior of $D$.

This eigenvalue has a remarkable isoperimetric property referred to as the
Faber--Krahn inequality; it is as follows:
\begin{equation}
\lambda_1 (D) \geq \pi j_{0,1}^2 / |D| ,
\label{1.2}
\end{equation}
where $|D|$ is the area of $D$ and $j_{0,1} = 2.4048...$ is the first zero of the
Bessel function $J_0$ (the notation of \cite{AS} is used for Bessel functions and
their zeroes). Equality is attained in \eqref{1.2} if and only if $D$ is a disc. In
other words, {\it among all homogeneous membranes of a given area, the circular one
has the lowest fundamental frequency} because $j_{0,1}^2$ is the lowest eigenvalue
for the unit disc. This inequality was conjectured in \cite{R}, pp. 339--340, on the
basis of numerical computations for simple domains and a variational argument for
nearly circular domains. Independent proofs of \eqref{1.2} were given by Faber
\cite{F} and Krahn \cite{K1}; the last author also proved its higher-dimensional
version \cite{K2}.

Various other versions of Faber--Krahn's inequality are discussed in \cite{AB} and
one of these versions belongs to Payne and Weinberger \cite{PW}. In terms of
\[ S_\alpha = \{ (r, \theta) : r \in (0, \infty) ; \theta \in (0, \pi / \alpha) \} ,
\]
where $\alpha \geq 1$ and $(r, \theta)$ is the polar coordinates system on the $(x,
y)$-plane, their result is as follows (see also \cite{B}, ch.~III, \S 2.3).

\begin{theorem}
Let $D \subset S_\alpha$ and let $\mathcal{I}_\alpha (D) = \int_D r^{2 \alpha +1}
\sin^2 \alpha \theta \, \D r \D \theta$ be fixed. Then
\begin{equation}
\lambda_1 (D) \geq \left[ 4 \pi^{-1} \alpha (\alpha +1) \mathcal{I}_\alpha (D)
\right]^{-1/(\alpha +1)} j_{\alpha,1}^2 ,
\label{1.3}
\end{equation}
where $j_{\alpha,1}$ is the first positive zero of the Bessel function $J_\alpha$.
Equality is attained when $D$ is a circular sector of angle $\pi / \alpha$.
\end{theorem}

Thus, a circular sector of angle $\pi / \alpha$ (in \cite{PW}, this number is
misprinted as $\alpha$) has the lowest fundamental eigenvalue among all domains
lying in $S_\alpha$ and having a given value of $\mathcal{I}_\alpha (D)$. 

It is natural to ask whether an analogue of this assertion is true for indented
domains; that is, for those located in a reflex angle (its radian measure is between
$\pi$ and $2 \pi$) or in the plane with an infinite straight cut. The aim of this
note is to show how to obtain the corresponding result by modifying considerations
in \cite{PW}. 

First, instead of $S_\alpha$ it is convenient to introduce
\[ R_\beta = \{ (r, \theta) : r \in (0, \infty) ; \ \theta \in (- \pi / \beta , 
\pi / \beta) \} \quad \mbox{for} \ \beta \in [1, 2] .
\]
The plane cut along the negative $x$-axis corresponds to $\beta = 1$ and for $\beta
\in (1, 2)$ one obtains the whole family of reflex angles centred at the origin;
finally the half-plane $\{ x > 0 \}$ corresponds to $\beta = 2$. Now we are in a
position to formulate the following result in addition to Theorem~1.1.

\begin{theorem}
Let $D \subset R_\beta$ and let $I_\beta (D) = \int_D r^{\beta +1} \cos^2
\frac{\beta \theta}{2} \, \D r \D \theta$ be fixed. Then
\begin{equation}
\lambda_1 (D) \geq \left[ \pi^{-1} \beta (\beta +2) I_\beta (D) \right]^{-2/(\beta
+2)} j_{\beta/2,1}^2 ,
\label{1.4}
\end{equation}
where $j_{\beta/2,1}$ is the first positive zero of the Bessel function
$J_{\beta/2}$. Equality is attained when $D$ is a circular sector of angle $2 \pi /
\beta$.
\end{theorem}

Thus, a circular sector of angle $2 \pi / \beta$ has the lowest fundamental
eigenvalue among all domains lying in $R_\beta$ and having a given value of $I_\beta
(D)$. It should be mentioned that \eqref{1.4} with $\beta = 2$ and \eqref{1.3} with
$\alpha = 1$ coincide, and so \eqref{1.4} extends \eqref{1.3} to reflex angles.

Both lower bounds \eqref{1.3} and \eqref{1.4} for particular domains depend on the
choice of the origin. In this regard, it is reasonable to cite Payne and Weinberger
\cite{PW}, p.~186. ``There appears to be no systematic method of determining the
origin to give the best lower bound. Experience and considerations of symmetry are
certainly helpful.''

\section{Auxiliary Lemma}

The following lemma provides the geometric inequality analogous to that proved by
Payne and Weinberger; see Lemma in \cite{PW}, \S 2.

\begin{lemma}
If $D \subset R_\beta$, then
\begin{equation}
\left[ \frac{\beta}{\pi} \int_{\partial D} r^\beta \cos^2 \frac{\beta \theta}{2} \,
\D s \right]^{(\beta + 2) / (\beta + 1)} \geq \pi^{-1} \beta (\beta +2) I_\beta
(D),
\label{2.1}
\end{equation}
and equality is attained when $D$ is a circular sector of angle $2 \pi / \beta$.
\end{lemma}

\begin{proof}
We just outline amendments to be made in the proof of Payne and Weinberger. First,
the mapping
\begin{eqnarray}
&& D \ni (x = r \cos \theta, y = r \sin \theta) \mapsto \nonumber \\ && \ \ \ \ \ \
\left( x_1 = r^{(\beta+1)/3} \cos \frac{\beta \theta}{2} , y_1 = r^{(\beta+1)/3}
\sin \frac{\beta \theta}{2} \right) \in D^*
\label{2.2}
\end{eqnarray}
must be applied instead of the transformation that appears in \cite{PW} under the
number (2.3). It is clear that \eqref{2.2} maps 
\[ R_\beta \supset D \mapsto D^* \subset \{ x_1 > 0 ; -\infty < y_1 < +\infty \} , \]
and since $\beta \geq 1$, the inequality
\[ \D s^2 \geq \frac{4 (\D x_1^2 + \D y_1^2)}{\beta^2 r^{2 (\beta - 2) / 3}}
\]
holds for the element of arc length $\D s$ measured along curves in the
$(x,y)$-plane. This implies that
\[ r^\beta \cos^2 \frac{\beta \theta}{2} \, \D s \geq 2 \beta^{-1} x^2 (\D x_1^2 + \D
y_1^2)^{1/2} ,
\]
where the integrand in \eqref{2.1} stands in the left-hand side.

The rest of lemma's proof literally repeats considerations in \cite{PW},
pp.~183--184, that follow formula (2.5) on p.~183. However, $x$ and $y$ must be
changed to $y_1$ and $x_1$ respectively. Indeed, $D^* \subset \{ x_1 > 0 ; -\infty <
y_1 < +\infty \}$ in the present case, whereas $D^*$ used in \cite{PW} lies in the
upper half-plane, and so $\sin \alpha \theta$ must be changed to $\cos \frac{\beta
\theta}{2}$.
\end{proof}

\section{Proof of Theorem 2}

The fundamental Dirichlet eigenvalue is characterized by the variational principle
based on the Rayleigh quotient
\begin{equation}
\lambda_1 (D) = \inf \frac{\int_D (w_x^2 + w_y^2) \, \D x \D y}{\int_D w^2 \, \D x
\D y} \, .
\label{3.1}
\end{equation}
It is sufficient to take this infimum over all $C^2 (D)$ functions which are
non-negative and vanish in a neighbourhood of $\partial D$. Since $D \subset
R_\beta$, any such trial function can be taken in the form
\[ w = v \, r^{\beta/2} \cos \frac{\beta \theta}{2} ,
\]
where $v$ belongs to the same class as $w$ itself. 

Let us consider the identity
\begin{eqnarray*}
&& \!\!\!\!\!\!\!\!\!\!\!\!\!\! \int_D \left[ (\phi \psi)_x + (\phi \psi)_y
\right]^2 \D x \D y = \int_D \phi^2 (\psi_x^2 + \psi_y^2) \, \D x \D y \\ && \ \ \ \
\ \ \ \ \ \ \ \ \ \ \ \ \ \ \ \ \ \ \ \ \ \ \, + \int_D \left[ \phi_x (\phi
\psi^2)_x + \phi_y (\phi \psi^2)_y \right] \D x \D y ,
\end{eqnarray*}
which holds for arbitrary $\phi$ and $\psi$. Putting $\phi = r^{\beta/2} \cos
\frac{\beta \theta}{2}$, $\psi = v$, and applying the divergence theorem to the last
integral, one obtains that this integral vanishes because $r^{\beta/2} \cos
\frac{\beta \theta}{2}$ is harmonic and $v$ is equal to zero on $\partial D$. Thus,
the equality
\[ \int_D (w_x^2 + w_y^2) \, \D x \D y = \int_D (v_x^2 + v_y^2) \, r^{\beta+1} \cos^2 
\frac{\beta \theta}{2} \, \D r \D \theta 
\]
is valid. Manipulating with the right-hand side integral in the same way as Payne
and Weinberger do with the right-hand side integral of their formula (3.4) (of
course, $\alpha$ must be changed to $\beta/2$ and $\sin$ to $\cos$), and using
inequality \eqref{2.1} instead of that proved in \cite{PW} (see Lemma on p.~183),
one arrives at the required inequality \eqref{1.4}.

\section{Examples}

In this section, we use subscripts to distinguish different domains.

\paragraph{Disc cut along a radius.} Let $D_{cd}$ be the disc of radius $\rho$
centred at the origin and cut along the negative $x$-axis; that is,
\[ D_{cd} = \{ (r, \theta) : r < \rho ; \ \theta \in (-\pi, \pi) \} . \]
In this case $\beta = 1$ and $J_{\beta/2} (t) = J_{1/2} (t) = \sqrt{2/(\pi t)} \sin
t$. Furthermore, equality is attained in formula \eqref{1.4}, according to which,
$\lambda_1 (D_{cd}) = (\pi / \rho)^2$ because $j_{1/2,1} = \pi$. The corresponding
eigenfunction is
\[ u_1 (D_{cd}) = J_{1/2} \left( \frac{\pi r}{\rho} \right) \cos \frac{\theta}{2} =
\sqrt{\frac{2 \rho} {\pi^2 r}} \sin \frac{\pi r}{\rho} \cos \frac{\theta}{2} \, .
\]
Thus, the first eigenvalue of a half-cut disc is $(\pi / j_{0,1})^2 = 1.7066...$
times larger than the first eigenvalue of the whole disc of the same radius.

\paragraph{Sector of an Annulus.} Let $D_{as} = \{ (r, \theta) : r \in (\rho_1 ,
\rho_2) ; \ \theta \in (-\pi / \beta, \pi / \beta) \}$ be the annular sector
centred at the origin. Then $\lambda_1 (D_{as}) = k^2$, where $k$ is the smallest
positive root of the equation
\[ J_{\beta / 2} (k \rho_1) Y_{\beta / 2} (k \rho_1) = J_{\beta / 2} (k \rho_2) 
Y_{\beta / 2} (k \rho_2) .
\]
A consequence of \eqref{1.4} is the lower bound $k \geq (\rho_2^{\beta+2} -
\rho_1^{\beta+2})^{-2 / (\beta+2)} j_{\beta/2,1}$ for this root. This bound is
similar to formula (3.27) in \cite{PW}.

\paragraph{Square cut along a half-midline.} Let us consider the domain
\[ D_1 = \{ (x, y) : |x| < 1 ; \ |y| < 1 ; \ \theta \neq \pm \pi) \} . \]
The exterior sides of this square are pairwise symmetric about the $x$ and $y$ axes,
its area is equal to 4 and it is cut along the negative $x$-axis. From \eqref{1.2},
it follows that
\begin{equation}
\lambda_1 (D_1) \geq \pi j_{0,1}^2 / 4 = 4.5420... \, .
\label{4.1}
\end{equation}
If the square has the same exterior sides as $D_1$ but no cut, then \eqref{1.2}
yields the same lower bound; that is,
\[ \lambda_1 (D_0) \geq \pi j_{0,1}^2 / 4 = 4.5420... \ \ \mbox{for} \ 
D_0 = \{ (x, y) : |x| < 1 ; \ |y| < 1 \} .
\]
Moreover, the last lower bound is less than 10$\%$ smaller than the exact value
\[ \lambda_1 (D_0) = \pi^2 / 2 = 4.9348... \, . 
\]
As in the case of discs with and without a cut, it is reasonable to expect that
$\lambda_1 (D_1) > \lambda_1 (D_0)$. Indeed, formula \eqref{1.4} with $\beta = 1$
and $j_{1/2,1} = \pi$ gives the following lower bound:
\begin{equation}
\lambda_1 (D_1) \geq \frac{\pi^{8/3}}{[ (\pi+1) \sqrt 2 + \log (1 + \sqrt 2)
]^{2/3}} = 5.9341... \, ,
\label{4.2}
\end{equation}
which is about 20$\%$ larger than the exact value for the uncut square $D_0$ and
substantially better than the Faber--Krahn bound \eqref{4.1}.

\paragraph{Square cut along a half-diagonal.} Let $D_2$ be as follows:
\[ \{ (x, y) : - \sqrt 2 / 2 < y-x < \sqrt 2 / 2 ; \ - \sqrt 2 / 2 < y+x < \sqrt 2 
/ 2 ; \ \theta \neq \pm \pi) \} .
\]
This square is also cut along the negative $x$-axis, but its vertices are located on
the $x$ and $y$ axes so that its area is equal to 4 like that of $D_0$ and $D_1$.
Therefore, the Faber--Krahn inequality \eqref{1.2} gives for $\lambda_1 (D_2)$ the
same lower bound as for $\lambda_1 (D_1)$ and $\lambda_1 (D_0)$; see \eqref{4.1}. It
occurs that formula \eqref{1.4} with $\beta = 1$ and $j_{1/2,1} = \pi$ gives the
following lower bound:
\begin{equation}
\lambda_1 (D_2) \geq \pi^2 / 2 = 4.9348... = \lambda_1 (D_0) .
\label{4.3}
\end{equation}
Comparing this lower bound with that following from the Faber--Krahn inequality, we
see that \eqref{4.3} is better. However, unlike the case of square cut along a
half-midline, \eqref{4.3} does not improve the bound natural form a physical point
of view.


{\small
\begin{thebibliography}{99}

\bibitem{AS} M. Abramowitz, I.\,A. Stegun, \textit{Handbook of Mathematical
Functions.} U.S. National Bureau of Standards, 1964.

\bibitem{AB} M.\,S. Ashbaugh, R.\,D. Benguria, \textit{Isoperimetric Inequalities
for Eigenvalues of the Laplacian.} Proc. Symp. Pure Math. \textbf{76}, Part 1
(2007), 105--139. Amer. Math. Soc.

\bibitem{B} C. Bandle, \textit{Isoperimetric Inequalities and Applications.} Pitman,
1980.

\bibitem{F} G. Faber, \textit{Beweis, dass unter allen homogenen Membranen von
gleicher Fl¨ache und gleicher Spannung die kreisf¨ormige den tiefsten Grundton
gibt.} Sitzungsber. math.-phys. Klasse Bayer. Akad. Wiss. (1923), 169--172.

\bibitem{K1} E. Krahn, \textit{\"Uber eine von Rayleigh formulierte
Minimaleigenschaft des Kreises.} Math. Ann. \textbf{94} (1925), 97--100.

\bibitem{K2} E. Krahn, \textit{\"Uber Minimaleigenschaften der Kugel in drei und
mehr Dimensionen.} Acta Comm. Univ. Tartu (Dorpat) \textbf{A9} (1926), 1--44.

\bibitem{KS} J.\,R. Kuttler, V.\,G. Sigillito, \textit{Eigenvalues of the Laplacian
in Two Dimensions.} SIAM Rev. \textbf{26} (1984), 163--193.

\bibitem{PW} L.\,E. Payne, H.\,F. Weinberger, \textit{A Faber--Krahn inequality for
wedge-like membranes.} J. Math. and Phys. (ZAMP) \textbf{39} (1960), 182--188.

\bibitem{R} J.\,W.\,S. Rayleigh, \textit{The Theory of Sound.} 2nd Edition in two
vols., Dover Publications, 1945.

\end{thebibliography}
}
\end{document}